\documentclass[11pt]{article}

\usepackage{times}
\usepackage{amsmath}
\usepackage{amsthm}
\usepackage{graphicx,epsfig}

\newcommand {\E}{\mathbf E}

\newcommand {\V}{\mathcal V}
\newcommand{\EE}{\epsilon}

\newcommand{\p}{\mathbf{p}}
\newtheorem{theorem}{Theorem}[section]

\newtheorem{lemma}[theorem]{Lemma}
\newtheorem{claim}[theorem]{Claim}
\newtheorem{corollary}[theorem]{Corollary}

\newtheorem{define}{Definition}
\newtheorem{observation}{Observation}

\advance \topmargin by -\headheight
\advance \topmargin by -\headsep
\evensidemargin \oddsidemargin
\begin{document}


\title{Budget Constrained Auctions with Heterogeneous Items}
\author{Sayan Bhattacharya\thanks{Department of Computer Science, Duke University, Durham NC 27708. Email: {\tt bsayan@cs.duke.edu}} \and Gagan Goel\thanks{College of Computing, Georgia Institute of Technology, Atlanta GA 30332. Email: {\tt gagang@cc.gatech.edu}. } \and Sreenivas Gollapudi\thanks{Microsoft Research Silicon Valley, Mountain View, CA 94043. Email: {\tt sreenig@research.microsoft.com}. } \and Kamesh Munagala\thanks{Department of Computer Science, Duke University, Durham NC 27708. Supported by an Alfred P. Sloan  Research Fellowship, and by NSF via CAREER award CCF-0745761 and grant  CNS-0540347. Email: {\tt kamesh@cs.duke.edu}}}
\date{}


\maketitle
\begin{abstract}
In this paper, we present  the first approximation algorithms for the problem of designing revenue optimal Bayesian incentive compatible auctions when there are multiple (heterogeneous) items and when bidders have arbitrary demand and budget constraints (and additive valuations). Our mechanisms are surprisingly simple: We show that a sequential all-pay mechanism is a $4$ approximation to the revenue of the optimal ex-interim truthful mechanism with a discrete type space for each bidder, where her valuations for different items can be correlated. We also show that a sequential posted price mechanism is a $O(1)$ approximation to the revenue of the optimal ex-post truthful mechanism when the type space of each bidder is a product distribution that satisfies the standard hazard rate condition. We further show a logarithmic approximation when the hazard rate condition is removed, and complete the picture by showing that achieving a sub-logarithmic approximation, even for regular distributions and one bidder, requires pricing bundles of items. Our results are based on formulating novel LP relaxations for these problems, and developing generic rounding schemes from first principles. 
\end{abstract}


  

\section{Introduction}
In several scenarios, such as the Google TV ad auction~\cite{nisan2} and the the FCC spectrum auctions~\cite{pino}, where auctions have been applied in the recent past, bidders are constrained by the amount of money they can spend. This leads to the study of auctions with budget-constrained bidders, which is the focus of this paper. The key difficulty with budgets is that the utility of a bidder is only quasi-linear as long as the price is below the budget constraint, but is $-\infty$ if the price exceeds the budget. As a consequence, well-known mechanisms such as the VCG mechanism are no longer directly applicable. Before proceeding further, we formally define our model.

\subsection{Our Model}
\label{sec:model} 
There are $m$ bidders and $n$ heterogeneous items. The {\em type} of a bidder is defined as the collection of her valuations for each item; we will assume throughout the paper that her valuations for different items are additive. Further, her type is private knowledge. We augment this basic model for bidder $i$ with two publicly known constraints: A {\em demand constraint} $n_i$ on the maximum number of items  she is willing to buy, and a {\em budget constraint} $B_i$ on the maximum total price she can afford to pay. Her utility is defined as follows: Suppose she gets a subset $A$ of items where $|A| \leq n_i$, and pays a total price $P_i$. If $v_{ij}$ denotes her valuation for item $j$, then her utility is given by $\sum_{j \in A} v_{ij} - P_i$ if $P_i \leq B_i$, and $-\infty$ otherwise. The goal of the auctioneer is to design a mechanism that maximizes its revenue, which is defined as $\sum_i P_i$. We require that the mechanism be {\em incentive compatible} in that no bidder gains in utility by misreporting her valuations, and {\em individually rational}, meaning that a bidder gets nonnegative utility by reporting the truth. Some of our results generalize to the case where the budget constraints are private knowledge.

There are two well-established ways of proceeding from here. In the adversarial setting, no assumptions are made on the valuations, while in the Bayesian setting, it is assumed the valuations are drawn from publicly known prior distributions. We take the latter Bayesian approach, which was pioneered by Myerson~\cite{myerson}. As mentioned above, in this setting, the bidders' private valuations are drawn from independent (but not necessarily identical) commonly known prior distributions. The auctioneer's goal is to maximize her expected revenue, where the expectation is over the prior distributions, and possible randomization introduced by the mechanism itself. In the Bayesian setting, the notions of truthfulness and individual rationality can be of two kinds:
\begin{description}
\item [Bayesian Incentive Compatible (BIC):] Here, each bidder's expected utility is maximized by truth-telling, where the expectation is over the valuations of the other bidders (which are drawn according to their priors) and the randomness introduced by the mechanism. In other words, truth-telling is optimal only in expectation. 
\item[Dominant Strategy Incentive Compatible (DISC):] In this setting, truth-telling is optimal for the bidder even if she knows the valuations of the other bidders and the random choices made by the mechanism. 
\end{description}
We note that traditionally, the terms BIC and DSIC only refer to randomness introduced by other bidders' valuations via the prior distributions. Therefore, to be very precise, the space of mechanisms is two-dimensional,  randomized versus universally truthful mechanisms on one dimension, and BIC versus DSIC mechanisms on the other. We collapse this space to randomized BIC versus universally truthful DSIC mechanisms, so as not to introduce additional terminology. 

\subsection{Our Results}
\label{sec:results} 
Our results depend on the structure of the type space of each bidder. In particular, the results depend on whether the distributions of the valuations of a bidder for different items are correlated or independent. We note that the correlations if any are {\em only} for valuations of one bidder for different items; the distributions of the valuations of {\em different} bidders are always assumed to be independent. All our results can be viewed as presenting simple characterizations of approximately optimal mechanisms in these contexts.

We will first consider the situation when valuations of a bidder for different items are arbitrarily correlated, so that the type space of a bidder is simply specified as a poly-bounded discrete distribution. In Section~\ref{sec:lot}, we present a simple {\em all-pay} mechanism: it charges each bidder a fixed price that depends only on its revealed type, while the allocation made to the bidder will be {\em random} depending on the types of other bidders. The resulting scheme is therefore only BIC; however, we show that it is a $4$ approximation to the revenue of the optimal BIC scheme (Theorem~\ref{thm:allpay}). 

\medskip
An all-pay auction is unrealistic in several situations, since a bidder is forced to participate even if she obtains negative utility when the auction concludes. A natural question to ask is whether there exists a DSIC mechanism with good revenue properties. However, it is an easy exercise to show that the problem of designing an optimal DSIC mechanism for one unit-demand bidder ($n_i = 1$) whose type can take one of $n$ arbitrarily correlated possibilities with equal probability, reduces to the problem of unlimited supply {\em envy-free pricing} with $n$ bidders~\cite{kempe}. For the envy free pricing problem, the best known is a logarithmic approximation, and there is strong evidence that a better approximation is not possible~\cite{Briest}. Furthermore, in the correlated valuation setting, there are examples showing an unbounded gap between the revenues of the optimal randomized (BIC) and deterministic (DSIC) mechanisms~\cite{shuchi2}. This significantly dims the possibility of a constant factor approximate DSIC mechanism with discrete correlated types.

In view of the above negative result, we assume the valuations of a bidder for different items follow independent distributions. This {\em product distribution} assumption is also used by Chawla {\em et. al.}~\cite{shuchi}, who consider the special case of a single bidder with unit demand. In this special case, it is easy to see that any DSIC mechanism is a posted price mechanism, and they show a constant factor approximation to its revenue by an elegant connection to Myerson's mechanism~\cite{myerson}. Independently of our work, Chawla {\em et al.}~\cite{shuchi1} extend their result to show a $O(1)$  approximation for multiple unit-demand bidders. However, this result also crucially require the unit-demand assumption. 

\medskip
In Section~\ref{sec:mon}, we show that for arbitrary demands and budgets, when the type of each bidder follows a product distribution where each dimension satisfies the standard monotone hazard rate (MHR), there is a constant factor approximation to the revenue of optimal BIC mechanism. Our mechanism is achieved by a DSIC mechanism with surprisingly simple structure: Consider the bidders sequentially in arbitrary order, and for each bidder, offer a subset of remaining items at a pre-computed price for each item; the bidder simply chooses the utility maximizing bundle of these items.  This posted price mechanism shows a constant factor gap between the revenues of the optimal BIC and DSIC mechanisms (Theorem~\ref{thm:optitem}), which is in sharp contrast with the corresponding negative result~\cite{shuchi2} when the type space of a bidder is correlated. The surprising aspect of our result is that since the optimal DSIC mechanism for general demands and budgets will price bundles of items,  there is {\em no a priori reason} to expect that the more simplistic scheme of posting prices for each item would be a constant approximation.

We show (Section~\ref{sec:gen}) that the MHR condition is indeed necessary for our item pricing result: If this condition is slightly relaxed to regular product distributions and one bidder, the optimal DSIC mechanism will price {\em bundles of items}, and will have a logarithmic gap against the revenue of the optimal item pricing scheme. We show that this gap is tight by showing the existence of a posted price mechanism achieving this approximation ratio (Theorem~\ref{thm:gen}). On a positive note, we conclude by proving that for regular distributions, if the space of feasible mechanisms is restricted to those that consider bidders in some adaptive order and post prices that may depend on the outcomes so far, there is a $O(1)$ approximation within this space, that considers bidders in an arbitrary but fixed order, and pre-computes the posted prices (Theorem~\ref{thm:final}).

\subsection{Related Work}
\label{sec:work}
The Bayesian setting is widely studied in the economics literature~\cite{BenoitK,pino,Che1,Che2,laffont,vincent,vohra,thas,wilson}. In this setting, the optimal (either BIC or DSIC) mechanism can always be computed by encoding the incentive compatibility constraints in an integer program and maximizing expected revenue. However, the number of variables (and constraints) in this IP is exponential in the number of bidders, there being variables for the allocations and prices for each scenario of revealed types. Therefore, the key difficulty in the Bayesian mechanism design case is {\em computational}: {\em Can the optimal (or approximately optimal) auction be efficiently computed and implemented?}

Much of the literature in economics considers the case where the auctioneer has one item (or multiple copies of the item). In the absence of budget constraints, Myerson~\cite{myerson} presents the characterization of any BIC mechanism in terms of expected allocation made to a bidder: This allocation must be monotone in the revealed valuation of the bidder, and furthermore, the expected price is given by applying the VCG calculation to the expected allocation. This yields a linear-time computable optimal revenue-maximizing auction that is both BIC and DSIC. The key issue with budget constraints is that the allocations need to be thresholded in order for the prices to be below the budgets~\cite{Che1,laffont,vohra}. However, even in this case, the optimal BIC auction follows from a polymatroid characterization that can be solved by the Ellipsoid algorithm and an all-pay condition~\cite{Border,vohra}. By {\em all-pay}, we mean that the bidder pays a fixed amount given his revealed type, regardless of the allocation made. This also yields a DSIC mechanism that is $O(1)$ approximation to optimal BIC revenue~\cite{sayan}, but the result holds only for homogeneous items. 

\medskip
An alternative line of work deals with the {\em adversarial setting}, where no distributional assumption is made on the bidders' private valuations. In this setting, the budget constrained auction problem is notorious mainly because standard auction concepts such as VCG, efficiency, and competitive equilibria do not directly carry over~\cite{nisan2}. Most previous results deal with the case of multiple units of a homogeneous good. In this setting, based on the random partitioning framework of Goldberg {\em et al.}~\cite{Goldberg04}, Borgs {\em et al.}~\cite{borgs} presented a truthful auction whose revenue is asymptotically within a constant factor of the optimal revenue  (see also~\cite{abrams}). If the focus is instead on optimizing social welfare, no non-trivial truthful mechanism can optimize social welfare~\cite{borgs}. Therefore, the focus has been on weaker notions than efficiency, such as Pareto-optimality, where no pair of agents (including the auctioneer) can simultaneously improve their utilities by trading with each other. Dobzinski {\em et al.}~\cite{nisan} present an ascending price auction based on the clinching auction of Ausubel~\cite{ausubel}, which they show to be the only Pareto-optimal auction in the public budget setting. This result was extended to the private budget setting by Bhattacharya {\em et al.}~\cite{sayan}; see ~\cite{ravi} for a related result. 

\medskip
Finally, several researchers have considered the behavior specific types of auctions, for instance, auctions that are sequential by item and second price within each item~\cite{BenoitK,elkind}, and ascending price auctions~\cite{ausubel2,pino}. The goal here is to analyze the improvement in revenue (or social welfare) by optimal sequencing, or to study incentive compatibility of commonly used ascending price mechanisms. However, analyzing the performance of sequential or ascending price auctions is difficult in general, and there is little known in terms of optimal mechanisms (or even approximately optimal mechanisms) in these settings.

\subsection{Our Techniques}
\label{sec:techniques}
If for every bidder, the valuations for different items are correlated (Section~\ref{sec:lot}), then the optimal BIC revenue can be bounded from above by a linear program ({\sc LP1}) that requires the incentive compatibility, voluntary participation, supply and demand constraints to hold {\em only in expectation}. We construct a BIC all pay auction (Figure~\ref{fig:allpay}) that basically implements a rounding scheme on the optimal solution to {\sc LP1}, loosing a constant factor in revenue (Theorem~\ref{thm:allpay}). This approach is based on the techniques used in~\cite{sayan}.

As mentioned before, Chawla {\em et. al.}~\cite{shuchi} consider the Bayesian unit-demand pricing problem. There are $n$ heterogeneous items, a single bidder with unit demand, and her valuations ($v_j$ for item $j \in [1, \ldots , n]$)  are drawn from independent distributions. They present an elegant pricing scheme that is a constant approximation to optimal revenue by upper bounding it using the revenue of Myerson's auction in the following setting: There is a single item, $n$ bidders, and valuation of each bidder $j$ follows the same distribution as that of $v_j$. However, this technique cannot be applied if the unit demand assumption is removed.  

In contrast, our approach to the independent valuations setting (Section~\ref{sec:seq}) does {\em not} require the unit-demand assumption, and is based on a novel LP relaxation ({\sc LPRev}) for the problem (Lemma~\ref{lem:crux}). Unlike the LP relaxation of Section~\ref{sec:lot} and perhaps surprisingly, {\sc LPRev} does not encode any incentive compatibility constraints, and our DSIC mechanism (Figure~\ref{fig:postprice}), that competes against this LP, is in fact a constant approximation to optimal BIC revenue. One limitation of our approach is that we have to (necessarily) assume the bidders' valuations satisfy montone hazard rate (MHR). In the process of proving our main result (Theorem~\ref{thm:optitem}), we describe a crucial property of MHR distributions (Lemma~\ref{lem:hazard}) that can be used to extend the type of results shown in~\cite{HR}. For example, in multi-item settings with only demand constraints, {\em posted price schemes generate revenue that is a constant factor of the optimal social welfare}, assuming MHR distributions (Corollary~\ref{cor1}). The LP formulations also generalize the stochastic matching setting in Chen {\em et al.}~\cite{chen}.

\section{Correlated Valuations:  BIC Mechanisms}
\label{sec:lot}
We first consider the problem of approximating the optimal Bayesian incentive compatible mechanism. There are $m$ bidders and $n$ indivisible items. The type $t$ of a bidder is a $n$-vector $\langle v_{1t}, v_{2t}, \ldots, v_{nt}\rangle$, which specifies her valuations for the items. For bidder $i$, we assume the types follow a  publicly known discrete distribution with polynomial support, where for $t = t_1, t_2, \ldots, t_K$, we have $f_{i}(t) = \Pr[\mbox{Type } = t]$. The distributions for different bidders are independent. Note that the above model allows the valuations of a bidder for different items to possibly be correlated.  The bidder can afford to pay at most a publicly known budget $B_i$, and is interested in at most $n_i \ge 1$ items.  (Our scheme extends to private budgets using the techniques in~\cite{sayan}.)   We will be interested in mechanisms that can be computed in time polynomial in the input size, {\em i.e.}, in $n, m$, and $K$.

For reported types $\vec{t}$, the auctioneer computes allocation of items $S_i(\vec{t})$ for each bidder $i$. It also computes prices $P_i(\vec{t})$. The subsets $S_i(\vec{t})$ of items need to be disjoint, and  $P_i(\vec{t}) \le B_i$. The  utility of a bidder for obtaining subset $S$ at price $P$ when her type is $t$ is $\sum_{j \in S} v_{jt}  - P$. A mechanism  should be incentive compatible in the following sense: For any bidder $i$, her expected utility, where the expectation is over the types of other bidders, is maximized if she reveals her true type.  The auctioneer is interested in designing an incentive compatible mechanism that maximizes revenue: $\E_{\vec{t}} \left[\sum_i  P_i(\vec{t}) \right]$.

We now show an all-pay auction mechanism that is a $4$-approximation to optimal revenue. Our LP relaxation and solution technique are inspired by the rounding scheme in~\cite{sayan} for multi-unit auctions.

\medskip
\noindent{\bf Linear Programming Relaxation. }
Throughout, we will denote bidders by $i$, items by $j$, and types by $t$. For any feasible mechanism, let $x_{ij}(t)$ denote the probability that bidder $i$ obtains item $j$ if her reported type is $t$.  Let $P_{i}(t)$ denote the expected price paid by bidder $i$ when she bids type $t$. We have the following LP:

\[ \mbox{Maximize } \ \ \ \sum_{i,t} f_{i}(t) P_{i}(t) \qquad \mbox{({\sc LP1})}  \]
\[ \begin{array}{rcll}
\sum_{i,t} f_{i}(t) x_{ij}(t) & \le & 1 & \forall j \\
\sum_{j} x_{ij}(t) & \le & n_i & \forall i, t \\
\sum_j v_{jt} x_{ij}(t) - P_{i}(t)  & \ge & \sum_j v_{jt} x_{ij}(s) - P_{i}(s)  & \forall i, t, s\\ 
\sum_j v_{jt} x_{ij}(t) - P_{i}(t)   & \ge & 0  & \forall i, t\\
x_{ij}(t) & \in & [0, 1] & \forall i, j, t\\
P_{i}(t) & \in & [0,B_i] & \forall i, t \\
\end{array} \]

The optimal Bayesian incentive compatible mechanism  is feasible for the above constraints: The first constraint simply encodes that in expectation, each item is assigned to at most one bidder; the second constraint encodes the demand. The third and fourth constraints encodes Bayesian incentive compatibility and individual rationality.   Therefore, the {\sc LP1} value is an upper bound on the expected revenue.

\newcommand{\G}{\mathcal{G}}

\begin{theorem}
\label{thm:allpay}
The All Pay Auction (Figure~\ref{fig:allpay}) is Bayesian Incentive Compatible, and its revenue is a $4$ approximation to the optimal BIC mechanism for heterogeneous items.
\end{theorem}

\begin{proof}
Note that $\sum_{i} X_{ij} \le \frac{1}{2}$ for all items $j$, since we scaled down the LP variables. Furthermore, $Z_{ij} = \prod_{i' < i} (1 - X_{ij}) \geq 1 - \sum_{i' < i} X_{ij} \ge  \frac{1}{2}$. This implies $\frac{1/2}{Z_{ij}} \leq 1$.
 
In Step 6, $J$ denotes the set of available items. When bidder $i$ is encountered in the predefined (but arbirtrary) ordering, the set $W_i$ is initialized to $\emptyset$.  In each group $\G_{it_i k}$, a single item $j$ is selected with probability $\tilde{x}_{ij}(t_i)$. If $j \in J$, it is added to $W_i$ with probability $\frac{1/2}{Z_{ij}}$ (which is at most 1), and regardless of the outcome of this random event, $j$ is removed from $J$. 

In the discussion below, all expectations are with respect the bids revealed by bidders $i' < i$ and the random choices made in constructing $W_i$.  Since the values revealed by bidders $i' < i$ are independent, item $j$ is available when bidder $i$ is considered with probability exactly $\prod_{i' < i} (1-X_{ij}) = Z_{ij}$. Therefore:
$$ \Pr[j \in W_i] = Z_{ij} \cdot \tilde{x}_{ij}(t_i) \cdot 1/(2Z_{ij})  = \tilde{x}_{ij}(t_i)/2 = x^*_{ij}(t_i)/4$$ 

By linearity of expectation,  the expected welfare of items allocated to bidder $i$ is precisely $\frac{1}{4} \sum_j v_{jt_i} x^*_{ij}(t_i)$. Since bidder $i$ pays a fixed price $P^*_i(t_i)/4$, both the expected welfare and expected price are scaled down by a factor exactly $4$ relative to the LP values  {\em regardless of her revealed type}. This preserves the incentive compatibility constraints in the LP, and makes the scheme be truthful and satisfy voluntary participation in  expectation over the types of other bidders and the randomness introduced by the mechanism. The theorem follows.
\end{proof}

\begin{figure}[htbp]
\centerline{\framebox{
\begin{minipage}{5.5in}
{\bf All Pay Auction}
\begin{tabbing}
\= 1. \ \ \  \=   Choose an arbitrary but fixed ordering of all  bidders and denote it by $1, \ldots, m$ \\ \\  
 \> 2. \> Find the optimal solution to {\sc LP1}, and  denote the variable values by $\{x^*_{ij}(t), P^*_i(t)\}$ \\ \\
  \> 3. \>  For all bidders $i$, items $j$, types $t$, let \\
  \> \> \ \ \ \ \ \ \ \ \ \ \ \  \= $\tilde{x}_{ij}(t) = x^*_{ij}(t) / 2$, \\ \\
  \> \> \> $X_{ij} = \sum_{t} f_{i}(t) \tilde{x}_{ij}(t)$, \\ \\
  \> \> \> $Z_{ij} = \prod_{i' < i} (1-X_{ij})$ \\ \\
    \> 4. \> \textsc{For} all bidders $i$, types $t$, \\
  \> \> \> Partition the items into $n_i$ disjoint groups s.t. in each group $\G_{itk}$, $\sum_{j \in \G_{itk}} \tilde{x}_{ij}(t) \le 1$ \\ \\
  \> 5. \> Collect the reported types of all bidders. Denote the reported type of bidder $i$ by $t_i$ \\ \\
  \> 6. \> Initialize $J$ to be the set of all items. \\
  \> \> \textsc{For} $i = 1, 2, \ldots m$, \\
  \> \> \> Initialize $W_i \leftarrow \emptyset$ \\
  \> \> \> \textsc{For} $k = 1 \mbox{ to } n_i$ \\
  \> \> \> \ \ \ \ \ \ \ \ \ \ \  \= Pick a single item $j \in \G_{i t_i k}$ w.p. $\tilde{x}_{ij}(t_i)$ \\
  \> \> \> \> \textsc{If} $j \in J$ \\
  \> \> \> \> \ \ \  \ \ \ \ \ \ \ \  \= $W_i \leftarrow W_i \cup \{j \}$ w.p. $\frac{1/2}{Z_{ij}}$ \\ 
  \> \> \> \> \> $J \leftarrow J \setminus \{j\}$  \\
  \> \> \> Bidder $i$ gets the (random) set $W_i$, and pays a (fixed) price $P^*_i(t_i) / 4$
  \end{tabbing}
\end{minipage}
}}
\caption{\label{fig:allpay} BIC Mechanism for correlated valuations} 
\end{figure}

We note that by replacing the objective function with $\sum_{i,j,t} f_{i}(t) v_{jt} x_{ij}(t)$, the resulting scheme is a $4$ approximation to the optimal expected social welfare (or any linear combination of revenue and welfare).

\section{Independent Valuations:  DSIC Mechanisms}
\label{sec:seq}
The All Pay Auction (Figure~\ref{fig:allpay}) satisfies the incentive compatibility constraints only in expectation over the valuations of the other bidders and randomness introduced by the mechanism. In practice, it is more desirable to consider mechanisms that satisfy these properties ex-post. As mentioned before (Section~\ref{sec:results}), in general, it is likely to be hard to approximate such mechanisms, and this hardness comes from the type space of a bidder being correlated.  However, we show that for the important special case of product distributions with the hazard rate condition, a constant factor approximation follows from sequential posted price mechanisms. Such a mechanism considers the bidders sequentially in arbitrary order, and for each bidder, posts a subset of remaining items at certain price for each item, and lets the bidder choose the utility maximizing bundle of these items. Such a scheme is ex-post incentive compatible by definition.

Formally, there are $n$ distinct indivisible goods each of unit quantity, and $m$ bidders. Bidder $i$ has publicly known budget $B_i$ on the total price she is willing to pay. Her utility for paying more than this price is $-\infty$; as long as the budget constraint is satisfied, the utility for any item is quasi-linear, {\em i.e.}, the valuation obtained minus price, thresholded below at zero. The valuations for different items are additive, and she can buy more than one item, but at most some $n_i \ge 1$ items. Her valuation for item $j$ is a positive integer-valued random variable $v_{ij} \in [1,L_{ij}]$ following publicly known distributions.  Though we assume the distributions are defined at positive integer values, this is w.l.o.g., since we can discretize continuous distributions in powers of $(1+\EE)$ and apply the same arguments. This also holds when the values taken by the random variables are not polynomially bounded.

\medskip
\noindent {\bf Key Assumption.} We will present ex-post incentive compatible mechanisms under the following key {\em product distribution} assumption: The valuations $v_{ij}$ for each $(i,j)$ follow independent distributions.

Under the above assumption we consider two cases: (1) The random variables $v_{ij}$ satisfy the monotone hazard rate condition (Section~\ref{sec:mon}), and (2) The random variables $v_{ij}$ are arbitrary integer valued variables over the domain $[1,L]$ (Section~\ref{sec:gen}). In the former case, we show that sequential posted prices are a $O(1)$ approximation to the optimal BIC mechanism (Theorem~\ref{thm:optitem}), whereas in the latter case, they are a $O(\log L)$ approximation, and doing better necessarily requires pricing bundles of items rather than individual items (Theorem~\ref{thm:gen}).

\medskip
\noindent {\bf Proof Sketch of Theorem~\ref{thm:optitem}.} 
First, observe that in order to enforce individual rationality, whenever item $j$ is allocated to bidder $i$, the expected revenue generated across edge $(i,j)$ is at most $\V_{ij} = \min(v_{ij}, B_i)$. Let $x_{ij}(r)$ denote the probability that bidder $i$ gets item $j$ when $\V_{ij} = r$. If we require the demand, budget and supply constraints to hold {\em only in expectation}, the revenue of any mechanism can be upper bounded by a linear program with variables $x_{ij}(r)$ (Lemma~\ref{lem:crux}). Interestingly, this linear program ({\sc LPRev}) does not encode incentive compatibility constraints. Thus, any DSIC scheme that is constant competitive against {\sc LPRev} will also be a $O(1)$ approximation to optimal BIC revenue. 
 
Note that if $v_{ij}$ satisfies MHR, then $\V_{ij}$ also satisfies MHR (Claim~\ref{claim1}). Next, we make the crucial observation that if the prior on a bidder's valuation is MHR, then with constant probability her virtual valuation is within a constant factor of her valuation (Lemma~\ref{lem:hazard}). Thus, replacing $\V_{ij}$ by the corresponding virtual valuation function $\varphi_{ij}$, we get a new linear program ({\sc LP2}) that is still a good approximation to {\sc LPRev} and hence to optimal revenue (Lemma~\ref{lem:LP2}). Using Myerson's characterization of virtual valuations~\cite{myerson}, we show that the optimal solution to {\sc LP2} has a nice structure (Lemma~\ref{lem:round}). Roughly speaking, it treats each edge $(i,j)$ as a separate bidder with valuation given by the random variable $\V_{ij}$. Finally, we employ an interesting rounding technique (Figure~\ref{fig:postprice}) on the optimal solution to {\sc LP2}, and get a DSIC item pricing scheme that looses a constant factor in revenue. Since we need to apply Markov's inequality during the rounding phase, we revisit the definition of $\V_{ij}$ and scale down the reported budget by a suitable factor (Definition~\ref{def:scale}). This implies {\sc LPRev} is an upper bound on optimal revenue up to a constant factor; rest of the proof remains the same. 

One issue with auctions where bidders are budget and demand constrained is that given posted prices for the items, the bidder needs to solve a two-dimensional {\sc Knapsack} problem to determine her optimal bundle, and this is {\sc NP-Hard} in general. However, this aspect does not change our results: Our upper bound {\sc LP2} is the best possible revenue regardless of what optimization the bidders perform, and furthermore, our analysis of the algorithm in Figure~\ref{fig:postprice} simply shows that with constant probability, the bidder solves a trivial knapsack instance where all items fit into the knapsack. Therefore, as long as the bidder uses any reasonable {\sc Knapsack} heuristic, our results hold.

\newcommand{\f}{\tilde{f}}

\subsection{Linear Programming Relaxation}
We first present a simple bound on the revenue of any BIC mechanism. In the ensuing discussion, we will denote bidders by $i$ and items by $j$. 

\begin{define}
\label{def:scale}
Define $\V_{ij} = \min(v_{ij}, B_i/4)$ and let its density function be $f_{ij}$.
\end{define}

\[ \mbox{Maximize } \ \ \ \sum_{i,j} \sum_r  r f_{ij}(r) x_{ij}(r) \qquad \mbox{({\sc LPRev})} \]
\[ \begin{array}{rcllr}
\sum_{j} \sum_r f_{ij}(r) x_{ij}(r)  & \le & n_i & \forall i  & \hfill (1)\\
\sum_{j} \sum_r r f_{ij}(r)  x_{ij}(r) & \le & B_i & \forall i & \hfill (2) \\
\sum_{i} \sum_r f_{ij}(r) x_{ij}(r)  & \le & 1 & \forall j  & \hfill (3)\\
x_{ij}(r) & \in & [0,1] & \forall i, j, r & \hfill (4)\\
\end{array}  \]

\begin{lemma}
\label{lem:crux}
 The optimal value of the program ({\sc LPRev}) is at least a factor $1/4$ of the revenue of the optimal BIC mechanism. Furthermore, in the LP solution, $x_{ij}(r)$ is a monotonically non-decreasing function of $r$.
\end{lemma}
\begin{proof}
Without any loss of generality, we assume that whenever a bidder is allocated a bundle of items, the total price she has to pay is distributed amongst the individual items obtained. Since the mechanism satisfies the individual rationality condition in expectation, it is easy to ensure that the expected price on a single item is never greater than the valuation for the item or the overall budget.  Let $x_{ij}(r)$ denote the expected amount of item $j$ bidder $i$ obtains when $\V_{ij} = r$, and let $p_{ij}(r)$ denote the expected price paid conditioned on obtaining the item. It is easy to see that $p_{ij}(r) \leq 4r$. The revenue of the optimal BIC mechanism can be relaxed as:

\[ \mbox{Maximize } \ \ \ \sum_{i,j} \sum_r  p_{ij}(r) f_{ij}(r) x_{ij}(r)  \]
\[ \begin{array}{rcllr}
\sum_{j} \sum_r f_{ij}(r) x_{ij}(r)  & \le & n_i & \forall i  & \hfill (1)\\
\sum_{j} \sum_r p_{ij}(r) f_{ij}(r)  x_{ij}(r) & \le & B_i & \forall i & \hfill (2) \\
\sum_{i} \sum_r f_{ij}(r) x_{ij}(r)  & \le & 1 & \forall j  & \hfill (3)\\
x_{ij}(r) & \in & [0,1] & \forall i, j, r & \hfill (4)\\
p_{ij}(r) & \in & [0,4r] & \forall i, j, r & \hfill (5)\\
\end{array}  \]
The above program is nonlinear. Scale $p_{ij}(r)$ down by a factor of $4$ so that $p_{ij}(r) \le r$. This preserves the constraints and loses a factor of $4$ in the objective.  Now, for all $i, j$, if $p_{ij}(r) < r$, increase $p_{ij}(r)$ and decrease $x_{ij}(r)$ while preserving their product until $p_{ij}(r)$ becomes equal to $r$. This yields the constraints of {\sc (LPRev)}, and preserves the objective; but now,  the objective becomes $ \sum_{i,j} \sum_r  r f_{ij}(r) x_{ij}(r)$. This shows {\sc (LPRev}) is a $4$-approximation to the revenue of the optimal BIC mechanism.

To show that the objective is maximized when the $x_{ij}(r)$ are monotonically non-decreasing in $r$, for any $(i,j)$, preserve $\sum_{r}  r f_{ij}(r) x_{ij}(r)$ by increasing $x_{ij}(s_2)$ and decreasing $x_{ij}(s_1)$ for $s_1 < s_2$. In this process, $\sum_r f_{ij}(r) x_{ij}(r)$ must decrease, preserving all the constraints, which implies the monotonicity.
\end{proof}

We note that because of the presence of budget constraints, {\sc (LPRev)} bounds only the expected revenue of the ex-post mechanism and {\em not} the expected social welfare, which can be larger by an unbounded amount. However, if the budget constraints are removed, the resulting LP also bounds the optimal social welfare. 

\subsection{Monotone Hazard Rates}
\label{sec:mon}
We will now present a constant factor approximation to the optimal BIC mechanism via sequential posted price schemes, assuming the random variables $v_{ij}$ satisfy the monotone hazard rate condition. 
\begin{define}
The distribution of $v_{ij}$ satisfies the monotone hazard rate  (MHR) condition if $\Pr[v_{ij} > r]/Pr[v_{ij} = r]$ is a non-increasing function of $r$. 
\end{define}

\begin{claim}
\label{claim1}
If $v$ is a random variable satisfying the monotone hazard rate condition, then $\min(v,a)$ satisfies the  MHR condition for any  integer $a \ge 1$.
\end{claim}

It follows that the  $\V_{ij}$ also satisfy the MHR condition.

\begin{define}[Myerson~\cite{myerson}]
\label{def:virt}
Let $G_{ij}(r)  = \Pr[\V_{ij} > r]$.  Define the {\em virtual valuation} as $\varphi_{ij}(r) = r - G_{ij}(r)/f_{ij}(r)$. This is said to be {\em regular} if it is a non-decreasing function of $r$.  
\end{define}

Clearly, MHR distributions are regular. We now present the crucial lemma for MHR distributions. 
\begin{lemma}
\label{lem:hazard}
Let $\phi(v)$ denote the virtual valuation function for a positive integer valued random variable $v \in [1,L]$ satisfying the MHR condition. Then, $\Pr[\phi(v) \ge \frac{v}{2}] \ge \frac{1}{e^2} $.
\end{lemma}
\begin{proof}
Let $f(t) = \Pr[v = t]$ and $G(t) = \Pr[v > t]$. Let $h(t) = f(t)/G(t)$ denote the hazard rate at $t$. Since $f$ satisfies MHR, $\phi(t) <  t/2$ iff $t \le k$ for some integer $1 \le k \le L$. Furthermore, $\phi(k) = k - 1/h(k) \le k/2$. This implies for all $t \le k$, $h(t) \le h(k) \le 2/k$. 

Therefore, $\Pr[\phi(t) \ge t/2] = \Pr[t > k]$. To bound the latter quantity, replace $f$ with a continuous-valued distribution $\hat{f}$ whose density in $[t,t+1)$ for any integer $t \in [1,L]$ is precisely $f(t)$. It is easy to see that this new density also satisfies the hazard rate condition. Let $\hat{G}(t) = 1 - \int_{q=1}^{t+1} \hat{f}(q) dq$, and let $\hat{h}(t) = \frac{\hat{f}(t)}{\hat{G}(t)}$. Thus,
\begin{equation*}
 \Pr[\phi(t) \ge t/2] = \Pr[t > k] = G(k)  = \hat{G}(k) 
 =   e^{-\int_{t=1}^{k+1} \hat{h}(t) dt} \ge  e^{-\int_{t=1}^{k+1} \frac{2}{k} dt} \ge  e^{-2}
\end{equation*}
 \end{proof}
 
\begin{observation}
Let $v^*_{ij} = \mbox{argmin}_r \{\varphi_{ij}(r) \ge \frac{r}{2}\}$. Then, Lemma~\ref{lem:hazard} and Claim~\ref{claim1} imply $\Pr[\V_{ij} \ge v^*_{ij}] \ge e^{-2}$.
\end{observation}

\subsubsection{Incorporating Virtual Valuations}  
Now consider the following linear program obtained from ({\sc LPRev}).
\[ \mbox{Maximize } \ \ \ \sum_{i,j} \sum_r  f_{ij}(r)  \varphi_{ij}(r) x_{ij}(r)  \qquad \mbox{({\sc LP2})} \]
\[ \begin{array}{rcllr}
\sum_{j} \sum_r f_{ij}(r) x_{ij}(r)  & \le & n_i & \forall i  & \hfill (1)\\
\sum_{j} \sum_r f_{ij}(r)  \varphi_{ij}(r) x_{ij}(r) & \le & B_i & \forall i & \hfill (2) \\
\sum_{i} \sum_r f_{ij}(r) x_{ij}(r)  & \le & 1 & \forall j  & \hfill (3)\\
x_{ij}(r) & \in & [0,1] & \forall i, j, r & \hfill (4) 
\end{array} \]
 
 \begin{lemma}
 \label{lem:LP2}
 The value of  {\sc (LP2)} is at least $\frac{1}{2e^2}$ times the value of {\sc (LPRev)}. 
 \end{lemma}
 \begin{proof}
In the optimal solution to {\sc (LPRev)}, since $\Pr[\V_{ij} \ge v^*_{ij}] \ge e^{-2}$ by Lemma~\ref{lem:hazard} and since $x_{ij}(r)$ is monotone in $r$, setting $x_{ij}(r) = 0$ for $r < v^*_{ij}$ loses a factor of at most $e^2$ in the objective value and preserves feasibility. Now replacing $r$ by $\varphi_{ij}(r) \in [r/2,r]$ for $r \ge v^*_{ij}$ loses another factor of $2$ by Lemma~\ref{lem:hazard}, and preserves feasibility.  This shows a feasible solution to ({\sc LP2}) of value at least $\frac{1}{2e^2}$ times that of {\sc (LPRev)}.
 \end{proof}
 
 \begin{lemma}
 \label{lem:round}
The optimal solution to {\sc (LP2)} is as follows: For each $(i,j)$, we have a convex combination of two solutions. The first (resp. second) solution has variables $y_{ij}(r)$ (resp. $z_{ij}(r)$), and value $r^*_{ij} \le B_i/4$ (resp. $s^*_{ij} = r^*_{ij}+1 \le B_i/4$) so that if $r < r^*_{ij}$, then $y_{ij} = 0$, else $y_{ij} = 1$. Similarly, if $r < s^*_{ij}$, then $z_{ij} = 0$, else $z_{ij} = 1$. If the first solution has weight $p_{ij}$ in the convex combination, and the second $1-p_{ij}$, then: 
\begin{equation}
\label{eq1}
 \sum_r f_{ij}(r)  \varphi_{ij}(r) x_{ij}(r) 
 = p_{ij} r^*_{ij} \Pr[\V_{ij} \ge r^*_{ij}] + (1-p_{ij}) s^*_{ij} \Pr[\V_{ij} \ge s^*_{ij}] 
\end{equation}
 \end{lemma}
 \begin{proof}
For each $(i,j)$ if $x_{ij}(r_1), x_{ij}(r_2) \in (0,1)$ for $r_1 < r_2 \le B_i/4$, then $x_{ij}(r_1)$ can be increased and $x_{ij}(r_2)$ can be decreased to preserve $\sum_r \varphi_{ij}(r) f_{ij}(r) x_{ij}(r)$. Since $\varphi_{ij}(r)$ is a monotonically non-decreasing function of $r$, this implies $\sum_r f_{ij}(r)   x_{ij}(r)$ must decrease in this process, preserving all constraints. When this process terminates, we must have $x_{ij}(r) = 0$ for all $r < r^*$ and $x_{ij}(r) = 1$ for all $r > r^*$ for some $r^*$.  This implies the optimal solution to {\sc (LP2)} can be written in the fashion implied by the lemma, with $r^*_{ij} = r^*$ and $p_{ij} = x_{ij}(r^*)$. By Myerson's characterization of virtual valuations~\cite{myerson}, we have:
$$   \sum_r f_{ij}(r) \varphi_{ij}(r) x_{ij}(r) = \sum_r \left(f_{ij}(r) \left(r x_{ij}(r) - \sum_{s=1}^{r-1} x_{ij}(s)  \right) \right) $$
Since $y_{ij}(r)$ is a $0/1$ step function, it  is easy to see that
$$ \sum_r \left(f_{ij}(r) \left(r y_{ij}(r) - \sum_{s=1}^{r-1} y_{ij}(s)  \right) \right) =  r^*_{ij} \Pr[\V_{ij} \ge r^*_{ij}]$$
A similar equality holds for $z_{ij}$, and taking a convex combination completes the proof of the lemma.
\end{proof}

\subsubsection{Posted Price Mechanism and Analysis} 
\label{sec:rounding}

The posted price auction is described in Figure~\ref{fig:postprice}. Note that though the scheme is randomized, since it is posted price, it is universally and ex-post truthful and individually rational.

\begin{figure}[htbp]
\centerline{\framebox{
\begin{minipage}{5.5in}
{\bf Posted Price Auction}
\begin{tabbing}
\= 1. \ \ \  \=   Choose an arbitrary but fixed ordering of all bidders and denote it by $1, \ldots, m$ \\ \\  
 \> 2. \> Solve {\sc LP2} \\ \\
 \> 3. \> \textsc{For} each $(i,j)$ \\ 
 \> \> \ \ \ \  \ \ \ \ \ \ \ \ \ \ \= Independently pick one of the two solutions in the convex combination (see Lemma~\ref{lem:round})\\
 \> \> \> with probability equal to its weight in the combination.  \\ \\
  \> 4. \> Let $\tilde{r}_{ij}$ denote the threshold in the chosen solution where the allocation function $\tilde{y}_{ij}$  jumps  \\  \> \> from zero to one. \\ \\
  \> 5. \> Initialize $J$ to be the set of all items. \\ 
  \>  \> \textsc{For} $i = 1, 2, \ldots, m$ \\ 
  \> \> \> Initialize $W_i \leftarrow \emptyset$ \\ 
  \> \> \> \textsc{For} each item $j \in J$ \\ 
  \> \> \> \ \ \ \ \ \ \ \ \ \ \ \=  $W_i \leftarrow W_i \cup \{j\}$ w.p. $\frac{1}{4}$ \\ 
  \> \> \> Only the set of items $W_i$ is offered to bidder $i$ \\
  \> \> \> Each $j \in W_i$ is offered at a price $\tilde{r}_{ij}$ \\ 
  \> \> \> Bidder $i$ buys a subset of items $S_i \subseteq W_i$ \\ 
  \> \> \> $J \leftarrow J \setminus S_i$
  \end{tabbing}
\end{minipage}
}}
\caption{\label{fig:postprice} DSIC Mechanism for independent valuations} 
\end{figure}

\begin{theorem}
\label{thm:optitem}
The  revenue of the posted price auction (Figure~\ref{fig:postprice}) is a $O(1)$ approximation to the optimal BIC mechanism, when the valuations of a bidder follow product distributions that satisfy the  MHR condition.
\end{theorem}
\begin{proof}
We show that for every pair $(i,j)$, the mechanism achieves a revenue that is within  a constant factor of $ \sum_r f_{ij}(r)  \varphi_{ij}(r) x_{ij}(r) $  in {\sc (LP2)} for the pair. First note that if a price is posted according to the mechanism in Figure~\ref{fig:postprice} just along edge $(i,j)$, then by Equation (\ref{eq1}), the expected revenue is simply:
\begin{equation*}
  p_{ij}  r^*_{ij} \Pr[\V_{ij} \ge r^*_{ij}] + (1-p_{ij}) s^*_{ij} \Pr[\V_{ij} \ge s^*_{ij}]  
 =  \sum_r f_{ij}(r)  \varphi_{ij}(r) x_{ij}(r)
 \end{equation*}
 
We will now show that even in the presence of other items and bidders, a constant factor of this revenue is extracted. For pair $(i,j)$, let $X_{ij}$ be the $0/1$ random variable denoting whether item $j$ is taken by bidder $i$, and let $P_{ij}$ denote the price at which it is taken (which is $0$ if the item is not taken by the bidder). These are random variables that depend on the prior distributions as well as the random choices made in choosing the structured solution and the mechanism. Since the probability that any $(i,j)$ is considered at all is $1/4$, the constraints of  {\sc (LP2)}, when scaled down by factor $4$, imply by linearity of expectation that for all $i,j$:
\begin{equation*}
 \E[\sum_{k \neq j}P_{ik}] \le B_i/4, \qquad  \E[\sum_{k \neq j} X_{ik}] \le n_i/4, \qquad
 \mbox{and} \qquad   \E[\sum_{m \neq i} X_{mj}] \le 1/4  
 \end{equation*}
Note that snce the valuation $v_{ij}$ is independent of other valuations, this implies the above statements hold regardless of $v_{ij}$. Now applying Markov's inequality, we have for all $i,j$:
\begin{equation*}
 \Pr[\sum_{k \neq j}P_{ik} \ge 3B_i/4]  \le 1/3, \qquad   \Pr[\sum_{k \neq j}X_{ik} \ge n_i ]  \le 1/4, \qquad  \mbox{and} \qquad  \Pr[\sum_{m \neq i} X_{mj} \ge 1] \le 1/4
\end{equation*}
By union bounds, this implies that with probability at least $1/6$, we must have: $\sum_{k \neq j}P_{ik} < 3B_i/4$, $\sum_{k \neq j}X_{ik} < n_i$, and $\sum_{m \neq i} X_{mj} = 0$. In this event, item $j$ is offered to bidder $i$ w.p. $1/4$ (using one of two random choices of the posted price from Lemma~\ref{lem:round}). Furthermore, in this event, the bidder must take the item if $v_{ij} \ge \tilde{r}_{ij}$, since  $\tilde{r}_{ij} \le B_i/4$ (so that the bidder has sufficient budget to purchase this item), and the bidder has not exhausted his demand $n_i$. Since the valuation $v_{ij}$ itself is independent of the event that the item $j$ is offered to bidder $i$,  this implies that with probability at least $1/6 \times 1/4$, the posted price mechanism obtains revenue $ \sum_r f_{ij}(r)  \varphi_{ij}(r) x_{ij}(r) $ along each $(i,j)$ (using the definition of the latter quantity from Equation (\ref{eq1})). By linearity of expectation over all $(i,j)$, we have a $O(1)$ approximation.
\end{proof}

\begin{corollary}
\label{cor1}
Under the assumptions of Theorem~\ref{thm:optitem}, if there are no budget constraints, then the revenue of the sequential posted price mechanism is a constant factor approximation to the optimal social welfare.
\end{corollary}
\begin{proof}
If the Budget Constraint (2)  is removed from {\sc (LPRev)}, it is an upper bound on the optimal social welfare that can be obtained by any mechanism. The rest of the analysis remains the same.
\end{proof}

\subsection{General Distributions}
\label{sec:gen}
The key point shown by Theorem~\ref{thm:optitem} is that the optimal revenue of an ex-post truthful mechanism is approximated to a constant factor by a simple sequential posted price auction, where the prices are posted for each item and each bidder. This crucially used the monotone hazard rate condition. The natural question to ask is whether such a simple mechanism is a good approximation for more general classes of distributions.  Note that even for one bidder, the optimal DSIC mechanism will in general price {\em bundles} of items. We show below that bundle pricing has a logarithmic advantage over item pricing even for regular distributions (see Definition~\ref{def:virt}), and that this gap is tight.

\begin{theorem}
\label{thm:gen}
Suppose $v_{ij} \in [1,L]$, are independent for different $(i,j)$ but do not necessarily satisfy the monotone hazard rate condition. Then, there is a $\Theta(\log L)$ gap between the revenues of the optimal posted price scheme and the optimal DSIC mechanism. The lower bound holds for regular distributions and one bidder.
\end{theorem}

\begin{proof}
To show the lower bound, consider the following scenario: There is only one bidder, $n$ items, and no budget or demand constraints. The valuations are $i.i.d.$ for each item $j$, and follow the common distribution with $G(r) = \Pr[v \ge r] =  \frac{1}{r}$, for $r = 1,2, \ldots, n$, so that $n = L$. This distribution is regular, since $\varphi(r) = -1$ for $r < n$, and $\varphi(n) = n-1$. We have $\E[v_{1j}] = H_n$ for all item $j$, and by Chernoff bounds, 
$$ \Pr[\sum_{j=1}^n v_{1j} \le (1-\epsilon) n H_n] \le e^{\frac{-n H_n \EE^2}{2n}} = o(1)$$
Therefore, a scheme that sets a price of $n H_n (1-\EE)$ for the bundle of $n$ items sells the bundle with probability $1-o(1)$, so that the expected revenue is $\Omega(n \log n)$. Any posted price scheme can extract a revenue of $\max_i i \cdot G(i) = 1$ from each item, so that the expected revenue from $n$ items is at most $n$. This shows a gap of $\Omega(\log n) = \Omega(\log L)$.

The upper bound follows from a standard scaling argument. We start with {\sc (LPRev)} which upper bounds the optimal achievable revenue. For each $(i,j)$, group the $r$ values in powers of $2$ so that there are $O(\log L)$ groups. Let group $\G_k$ denote the interval $[2^k, 2^{k+1})$. 

The first step is to solve {\sc (LPRev)}. Let its optimal value be $OPT$; note that this is a $4$ approximation to the optimal ex-post incentive compatible scheme. Now, for each $(i,j)$ perform the following rounding. Let 
\begin{equation*}
\begin{split} 
& k_{ij} = \mbox{argmax}_k \sum_{r \in \G_k} r f_{ij}(r) x_{ij}(r) \\
& r^*_{ij} = 2^{k_{ij}} \\
& x^*_{ij} = \frac{\sum_{r \in \G_k} f_{ij}(r) x_{ij}(r)}{\sum_{r \in \G_k} f_{ij}(r)} \\
& Q_{ij} = \Pr[\tilde{V}_{ij} \ge r^*_{ij}]
\end{split}
\end{equation*}
Since $x_{ij}(r)$ is monotone in $r$, it is easy to observe that:
\begin{enumerate}
\item $x^*_{ij} \le x_{ij}(r)$ for $r \in \G_k$ for $k > k^*_{ij}$. This implies $x^*_{ij} Q_{ij} \le \sum_r f_{ij}(r) x_{ij}(r)$.
\item $\sum_{r \in \G_{k_{ij}}} r f_{ij}(r) x_{ij}(r) \le 2 r^*_{ij} x^*_{ij} \Pr[r \in \G_{k_{ij}}]$. 
\item $\max_k \sum_{r \in \G_k} r f_{ij}(r) x_{ij}(r) =  \Omega(1/\log L)  \sum_r  r f_{ij}(r) x_{ij}(r) $. 
\item Therefore, $r^*_{ij} x^*_{ij} Q_{ij} = \Omega(1/\log L)  \sum_r  r f_{ij}(r) x_{ij}(r) $. Furthermore, $r^*_{ij} x^*_{ij} Q_{ij}  \le  \sum_r  r f_{ij}(r) x_{ij}(r) $.
\end{enumerate}

The above directly implies $\sum_{i,j} r^*_{ij} x^*_{ij} Q_{ij}  = \Omega(OPT/\log L)$, and the following constraints:
\[\begin{array}{rcll}
 \sum_{j}  x^*_{ij} Q_{ij}  & \le &  n_i &  \forall  i  \\
 \sum_{j} r^*_{ij} x^*_{ij} Q_{ij}  & \le &  B_i  &  \forall  i  \\
 \sum_{i} x^*_{ij} Q_{ij}  & \le &  1  & \forall   j  \\
 x^*_{ij}  & \in &  [0,1] & \forall i, j \\
  r^*_{ij} & \in &  [0,B_i/4] &\forall i,j 
  \end{array} \]
The final mechanism chooses an arbitrary ordering of bidders. When bidder $i$ is encountered, let $S_i$ denote the subset of items that have survived so far. For each item $j \in S_i$, with probability $x^*_{ij}/4$, post price $r^*_{ij}$, and with the remaining probability skip this item. This process is done simultaneously for all items, and the resulting items and prices are simultaneously offered to the bidder. Note that though the scheme is randomized, since it is posted price, it is universally and ex-post truthful and individually rational. Using  the same proof as Theorem~\ref{thm:optitem} now shows that this posted price scheme extracts revenue that is a constant factor of $\sum_{i,j} r^*_{ij} x^*_{ij} Q_{ij}$. This shows a $O(\log L)$ approximation.
\end{proof}

\subsubsection{Adaptive Posted Price Schemes}  The above implies we cannot generalize the result in Section~\ref{sec:mon} even to the class of regular distributions (see Definition~\ref{def:virt}). This is because an $\omega(1)$ gap is introduced in going from {\sc (LPRev)} to {\sc (LP2)}. However, we can show that if the space of mechanisms is restricted, {\sc (LP2)} itself is a good relaxation to the optimal revenue in this space.

In particular, suppose we restrict the space of mechanisms to be those that are posted price, and sequential by bidder: Depending on the subset of items and bidders left, the mechanism adaptively chooses the next bidder and posts  prices for a subset of the remaining items. The bidder being a utility maximizer, solves a knapsack problem to choose the optimal subset of items.  Once this bidder is dealt with, the mechanism again adaptively chooses the next bidder depending on the acceptance strategy of this bidder. Assuming the valuations of a bidder follow a product distribution, and $v_{ij}$ follows a regular distribution, we will show a $O(1)$ approximation to the optimal mechanism.  

Since the optimization problem the bidders need to solve is a two-dimensional {\sc Knapsack} problem which is {\sc NP-Hard}, we we assume bidders are {\em monotone optimizers} in the following sense: For any edge $(i,j)$, if all valuations except $v_{ij}$ are fixed, the quantity of item $j$ taken by bidder $i$ is monotonically non-decreasing in $v_{ij}$. This is true if the bidder solves {\sc Knapsack} optimally. As with the posted price scheme in the previous section, our algorithm itself will allow arbitrary {\sc Knapsack} heuristics as long as a bidder chooses all items if she is not constrained by either demand or budget.

\renewcommand{\r}{\tilde{r}}
\newcommand{\x}{\tilde{x}}
\renewcommand{\p}{\tilde{p}}

\begin{lemma}
{\sc (LP2)} is a $4$-approximation to the revenue of the optimal  adaptive posted price scheme assuming bidders are monotone optimizers.
\end{lemma}
\begin{proof}
In the discussion below, recall the definition $\V_{ij} = \min(v_{ij}, B_i/4)$  from Definition~\ref{def:scale}, and observe that its density function $f_{ij}$ satisfies the regularity condition. For any adaptive mechanism, let $x_{ij}(r)$ denote the expected allocation of item $j$ to bidder $i$ conditioned on $\V_{ij} = r$, and let $p_{ij}(r)$ denote the corresponding expected price paid scaled down by a factor of $4$. We will show that the LP relaxation that is a $4$-approximation to maximizing expected revenue is the following:
\[ \mbox{Maximize } \ \ \ \sum_{i,j} \sum_r  f_{ij}(r) p_{ij}(r)  \qquad \mbox{({\sc LPSeq})} \]
\[ \begin{array}{rcllr}
\sum_{j} \sum_r f_{ij}(r) x_{ij}(r)  & \le & n_i & \forall i & \hfill (1)\\
\sum_{j} \sum_r f_{ij}(r)  p_{ij}(r)  & \le & B_i & \forall i & \hfill (2) \\
\sum_{i} \sum_r f_{ij}(r) x_{ij}(r)  & \le & 1 & \forall j  & \hfill (3)\\
r x_{ij}(r) - \sum_{s=1}^{r-1} x_{ij}(s) & \ge & p_{ij}(r)  & \forall i, j, r & \hfill (4) \\
x_{ij}(r) & \ge & x_{ij}(s)  & \forall i,j, r \ge s \ge 0 & \hfill (5)\\ 
p_{ij}(r) & \ge & p_{ij}(s)  & \forall i,j, r \ge s \ge 0 & \hfill (6)\\ 
x_{ij}(r) & \in & [0,1] & \forall i, j, r & \hfill (7)\\
\end{array}  \]
To see this, fix all valuations except $v_{ij}$. If the mechanism considers item $j$ for bidder $i$,  then it posts a unique price $p^* \le B_i$ for $(i,j)$. Let $X_{ij}(r) \in [0,1]$ denote the expected quantity taken when $\V_{ij} = r$, and let $P_{ij}(r) = p^* X_{ij}(r)$ denote the expected price paid. Note that as long as $r < B_i / 4$, $v_{ij} = r$ and $X_{ij}(r)$ is either $0$ or $1$. However, at $r = B_i / 4$, $v_{ij}$ is not unique and $X_{ij}(r)$ can take a fractional value. Since the bidder is a utility maximizer, $X_{ij}(r)$ (and hence $P_{ij}(r)$) are monotone functions of $r$. Furthermore, it is easy to check that $r X_{ij}(r) - \sum_{s=1}^{r-1} X_{ij}(s)  \ge  p^* X_{ij}(r)/4 = P_{ij}(r)/4$. Defining $x_{ij}(r) = \E[X_{ij}(r)]$ and $p_{ij}(r) = \E[P_{ij}(r)]/4$, where the expectation is over the remaining valuations, we see that all constraints in the above LP hold. 

 Consider the optimal solution to {\sc (LPSeq)}, and focus on edge $(i,j)$. Let $g_{ij}(r)  = r x_{ij}(r) - \sum_{s=1}^{r-1} x_{ij}(s) $. This is a monotone function of $r$, and so is $p_{ij}(r) \le g_{ij}(r)$. Consider the {\em lexicographically maximal} solution where it cannot be the case that  $p_{ij}(r)$ can be increased and some $p_{ij}(s)$ for $s > r$ can be decreased preserving all constraints and the objective. Let $r_0$ be the largest integer such that for all $r < r_0$, $g_{ij}(r) = p_{ij}(r)$. It is easy to see that for all $r \ge r_0$, we have $p_{ij}(r) = p_{ij}(r_0)$, else the lexicographic maximality property is violated. Now, reduce $x_{ij}(r_0)$ until $p_{ij}(r_0) =  g_{ij}(r_0)$. Observe that if $x_{ij}(r)$ is set to $x_{ij}(r_0)$ for all $r > r_0$, all constraints are preserved and the objective is unchanged.  This also ensures that $p_{ij}(r) =  g_{ij}(r)$ for all $r \ge 0$. This makes Constraints (4) and (6) irrelevant since the former is tight and implies the latter via Constraint (5). Thus, {\sc (LPSeq)} reduces to:
\[ \mbox{Maximize } \ \ \ \sum_{i,j} \sum_r  f_{ij}(r) \left(r x_{ij}(r) - \sum_{s=0}^{r-1} x_{ij}(s) \right)  \]
\[ \begin{array}{rcllr}
\sum_{j} \sum_r f_{ij}(r) x_{ij}(r)  & \le & n_i & \forall i\\
\sum_{j} \sum_r f_{ij}(r)   \left(r x_{ij}(r) - \sum_{s=1}^{r-1} x_{ij}(s) \right) & \le & B_i & \forall i \\
\sum_{i} \sum_r f_{ij}(r) x_{ij}(r)  & \le & 1 & \forall j\\
x_{ij}(r) & \in & [0,1] & \forall i, j, r\\
\end{array}  \]

Now, by Myerson's characterization~\cite{myerson}, we have:
$$   \sum_r \left(f_{ij}(r) \left(r x_{ij}(r) - \sum_{s=1}^{r-1} x_{ij}(s)  \right) \right)  = \sum_r f_{ij}(r) \varphi_{ij}(r) x_{ij}(r) $$
Plugging this formula into the objective and the second constraint completes the proof. 
\end{proof}
The final mechanism and analysis are the same as in Section~\ref{sec:mon}: Solve {\sc (LP2)}, decompose it into a convex combination of posted prices per edge, and sequentially post these prices for every bidder. To complete the analysis, note that Lemma~\ref{lem:round} only requires that the distribution be regular. This shows the following theorem:

\begin{theorem}
\label{thm:final}
There is a $O(1)$ approximation to the revenue of the optimal adaptive posted price scheme, when the types follow product distributions, and for each $(i,j)$, the distribution of valuation is regular.
\end{theorem}

\noindent{\bf Acknowledgments:} We thank Shuchi Chawla, Vincent Conitzer, and Peng Shi for helpful discussions.

\end{document}